\documentclass[conference]{IEEEtran}
\IEEEoverridecommandlockouts
\usepackage{cite}
\usepackage{amsmath,amssymb,amsfonts}
\usepackage{amsthm}
\usepackage{algorithmic}
\usepackage{graphicx}
\usepackage{subfigure}
\usepackage{caption}
\usepackage{geometry}
\usepackage{textcomp}
\usepackage{xcolor}
\newtheorem{theorem}{Theorem}

\newtheorem{lemma}[theorem]{Lemma}

\newtheorem{corollary}[theorem]{Corollary}
\newtheorem{defn}{Definition}

\def\BibTeX{{\rm B\kern-.05em{\sc i\kern-.025em b}\kern-.08em
    T\kern-.1667em\lower.7ex\hbox{E}\kern-.125emX}}

\long\def\symbolfootnote[#1]#2{\begingroup
\def\thefootnote{\fnsymbol{footnote}}\footnote[#1]{#2}\endgroup}
\IEEEoverridecommandlockouts

\begin{document}

\title{Grid-like Error-Correcting Codes for Matrix Multiplication with Better Correcting Capability}

\author{Hao Shi$^\dagger$, Zhengyi Jiang$^\dagger$$^\ddagger$, Zhongyi Huang$^\dagger$,
Bo Bai$^\ddagger$, Gong Zhang$^\ddagger$, and Hanxu Hou$^\ddagger$\\
$^\dagger$ Department of Mathematics Sciences, Tsinghua University, Beijing, China \\
$^\ddagger$ Theory Lab, Central Research Institute, 2012 Labs, Huawei Tech. Co. Ltd., Hong Kong SAR
}
\maketitle
\begin{abstract}
\symbolfootnote[0]{
This work was partially supported by the National Key R\&D Program of
China (No. 2020YFA0712300), the National Natural Science Foundation of China (No. 62371411, 61901115, 12025104) and the Key Area Research and Development
Program of Guangdong Province under Grant 2020B0101110003. And this work was supported in part by the National Natural Science Foundation of China under Grant 62371411.
{\em Corresponding author: Hanxu Hou}.}
Matrix multiplication over the real field constitutes a foundational operation in the training of deep learning models, serving as a computational cornerstone for both forward and backward propagation processes. However, the presence of silent data corruption (SDC) in large-scale distributed training environments poses a significant threat to model convergence and predictive accuracy, particularly when such errors manifest during matrix multiplication. Due to their transient and non-intrusive nature, these errors often evade detection, allowing them to propagate and accumulate over time, ultimately leading to substantial degradation in model performance. In this paper, we introduce a novel error-correcting coding framework specifically tailored for matrix multiplication operations. Our proposed framework is designed to detect and correct multiple computational errors that may arise during the execution of matrix products. By leveraging a grid-based structural encoding scheme, our approach enhances error localization and correction capabilities across all participating matrices, thereby significantly improving the fault tolerance of the computation. Experimental results demonstrate that our method achieves deterministic correction of up to two erroneous symbols distributed across three matrices with 100\% reliability, while incurring only a 24\% overhead in computational time on GPU architectures. Furthermore, we provide a rigorous theoretical analysis of the error-correction properties inherent to our coding scheme, establishing its correctness and robustness under well-defined fault models.
	\end{abstract}

\section{Introduction}
Matrix multiplication over the real field, a fundamental operation in linear algebra, serves as a cornerstone of deep learning model training \cite{Goodfellow-et-al-2016}. This operation enables the efficient computation of linear transformations that form the backbone of neural network architectures. Beyond its role in forward propagation—where input data is processed through successive layers—matrix multiplication is equally critical in back-propagation \cite{lecun1998gradient}, facilitating the gradient computations necessary for optimization. Consequently, the computational efficiency and speed of matrix multiplication directly influence the performance and scalability of modern deep learning algorithms.

In large-scale distributed training environments, silent data corruption (SDC) in hardware components poses a critical challenge to model convergence \cite{10.1145/3579371.3589105, dixit2022detecting, 9366780}. Empirical studies indicate that a substantial fraction of such errors occur during matrix multiplication operations. Unlike catastrophic failures, SDC events are inherently stochastic and evade conventional detection mechanisms \cite{gemini}, often allowing training processes to proceed to completion without explicit crashes. This stealthy behavior arises from the non-intrusive nature of SDC, which manifests as computational errors without triggering system-level anomalies.
The insidious impact of SDC is particularly pronounced in computationally intensive operations such as matrix multiplication, where even minor numerical deviations can propagate and amplify across iterations, ultimately compromising model performance. Robust error detection and correction methodologies are therefore essential to safeguard training stability and reliability \cite{llama3blog}, ensuring both the accuracy of learned representations and the generalization capabilities of deep neural networks.

Algorithm-Based Fault Tolerance (ABFT) \cite{1676475} has emerged as a powerful paradigm for mitigating silent data corruption (SDC) in large-scale computing systems. This approach embeds fault detection and correction mechanisms directly into computational algorithms, enabling real-time verification with minimal performance overhead. Through the strategic incorporation of checksum \cite{braun2014abft} or weighted checksum \cite{jou1984fault} techniques within matrix multiplication operations, ABFT can efficiently detect computational anomalies indicative of SDC events. Upon error detection, the system can either perform immediate correction or mark the compromised data for recomputation, thereby preserving computational integrity.
While current ABFT implementations demonstrate effectiveness in handling single-error scenarios, they exhibit significant limitations. Existing solutions are primarily constrained to detecting and correcting single errors occurring exclusively in the output matrix. This restricted capability leaves systems vulnerable to multiple concurrent errors or even single errors occurring in input matrices, representing a critical gap in current fault tolerance methodologies.

Recent advances in fault-tolerant computing have introduced analog error-correcting codes \cite{8849843, roth2020analog, wei2024multiple, jiang2024analog} as a specialized solution for error correction in the real domain. These coding schemes are specifically designed to address the unique challenges of approximate matrix multiplication in the presence of hardware-induced imperfections. Unlike traditional digital error correction methods, analog error-correcting codes employ sophisticated mechanisms to: (1) detect and identify outlying errors, (2) distinguish these from acceptable computational noise, and (3) provide robust error estimation while maintaining tight error bounds. 

In this paper, we propose a novel error-correcting coding framework for matrix multiplication based on ABFT with better correct capability during the computation process. 
By using a grid-like structure, we can accurately detect and correct multiple errors that occurred in the three matrices during matrix multiplication,  as well as maintain the rapid correction of a single error. 
We explored and analyzed the maximum error-correcting capacity of our framework.
Moreover, in the experiment, our framework can support 100\% detection and correction of the six types of data silence corruptions we proposed (the last three have not been corrected by any algorithm before). In terms of execution time on GPU, not only is our performance basically on par with that of the checksum algorithm for the first three errors, but we can also correct two symbols at only 24\% more time cost for the last three errors.

\section{PRELIMINARIES}
\subsection{Matrix Multiplication}

Denote $\{1,2,\ldots,\ell\}$ as $[\ell]$, where $\ell$ is an integer. 
For a general matrix multiplication operation, we consider the following model:
the left matrix $A = (a_{i, \ell})_{i\in [n] , \ell \in [k]}$, the right matrix $B = (b_{\ell, j})_{\ell \in [k], j \in [m]}$ 
and the output matrix $C = (c_{i, j})_{i \in [n], j \in [m]}$,
\begin{eqnarray*}
   C &=& AB, \\
   c_{i, j} &=& \sum_{\ell = 1}^{k} a_{i, \ell} b_{\ell, j}.
\end{eqnarray*}

During the process of matrix multiplication, three common types of errors can occur:

\begin{enumerate}
  \item[(E1)] An error in a symbol of the left matrix $A$;
  \item[(E2)] An error in a symbol of the right matrix $B$;
  \item[(E3)] An error during the computation of a symbol in matrix $C$, which is the product of matrices $A$ and $B$.
\end{enumerate}
These are also the three types of errors addressed in this paper, while previous work (such as ABFT algorithms) only focuses on error (E3). Notice that, 
if error (E1) occurs, i.e., a symbol $a_{i,\ell}$ in matrix $A$ is corrupted, it affects all symbols in the $i$-th row of the product matrix $C$ due to the row-wise computation pattern of matrix multiplication. Similarly, if error (E2) occurs, 
corruption of a symbol $b_{\ell,j}$ in matrix $B$ propagates to all symbols in the $j$-th column of $C$, reflecting the column-wise nature of the computation. 
Thus, all three types of errors can ultimately be traced back to errors in 
the symbols of matrix $C$.
A critical observation is that errors occurring in matrices $A$ and $B$ must logically precede the computation of parity symbols. If errors were instead assumed to affect $A$ and $B$ after parity symbol generation, the parity symbols themselves would inherently incorporate these errors, thereby invalidating their corrective function.

\subsection{Analog Error-Correcting Codes}
Analog error-correcting codes \cite{8849843, roth2020analog, wei2024multiple, jiang2024analog} are coding schemes that provide the 
ability to locate computational errors while using analog devices for approximate real matrix multiplication.

Consider that the output matrix $C$ may be different from the ideal computation
$C = AB \in \mathbb{R}^{n \times m}$ due to the effect of two events,
\begin{equation}\label{eq:analog}
C = AB + \boldsymbol{\varepsilon} + \boldsymbol{e},
\end{equation}
where $\boldsymbol{\varepsilon} = (\varepsilon_{i,j})_{i=1,\ldots,n}^{j=1,\ldots,m} \in \mathbb{R}^{n \times m}$, $\boldsymbol{e} = (e_{i,j})_{i=1,\ldots,n}^{j=1,\ldots,m} \in \mathbb{R}^{n \times m}$ denote the error matrices and each symbol $\varepsilon_{i,j}$ is bounded within $[-\delta, \delta]$ for some predefined $\delta > 0$. These bounded errors represent tolerable computational inaccuracies or circuit-level noise. In contrast, the symbols $e_{i,j}$ correspond to substantial errors that may originate from severe hardware faults such as stuck cells or short-circuit conditions in the memory array.

The primary objective of analog error-correcting codes is to develop an encoding scheme capable of: (i) detecting nonzero symbols of $\boldsymbol{e}$ exceeding the threshold interval $[-\Delta, \Delta]$, where $\Delta$ is minimized subject to system constraints; and (ii) estimating the magnitudes of these significant errors, under the assumption that their cardinality remains below specified bounds. In this paper, we set $\delta = \Delta$.

\section{Algorithm Framework}

\begin{figure*}
\label{fig:Generalframework}
  \centering
  \subfigure{
    \includegraphics[width=0.9\textwidth]{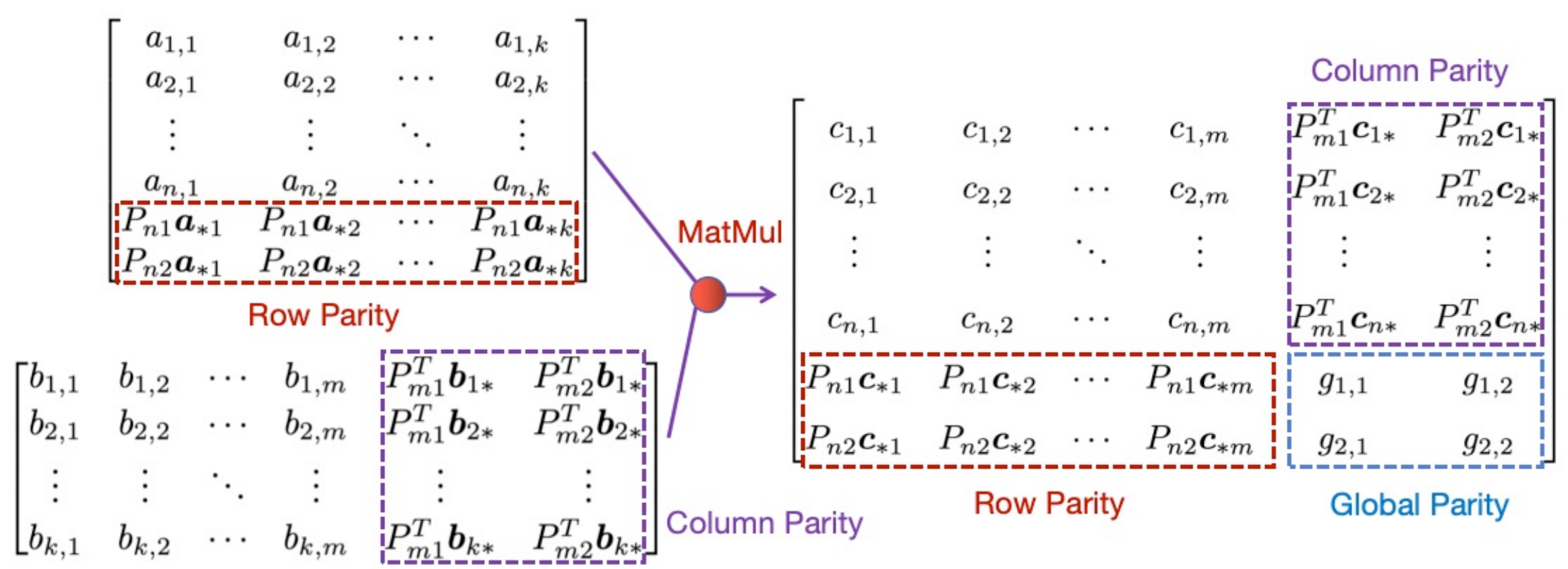}
  }
  \caption{The general framework $\mathcal{M}(n+2, k, m+2)$ of our error-correcting coding scheme for matrix multiplication.}
\end{figure*}

In this section, we introduce our framework for error-correcting in matrix multiplication. Define that 
\begin{eqnarray*}
  && \boldsymbol{a}_{*i} = (a_{1,i}, a_{2,i}, \ldots, a_{n,i})^T, \quad i \in [k], \\
  && \boldsymbol{b}_{i*} = (b_{i,1}, b_{i,2}, \ldots, b_{i,m}), \quad i \in [k], \\
  && \boldsymbol{c}_{*i} = (c_{1,i}, c_{2,i}, \ldots, c_{n,i})^T, \quad i \in [m],\\
  && \boldsymbol{c}_{i*} = (c_{i,1}, c_{i,2}, \ldots, c_{i,m}), \quad i \in [n].\\
\end{eqnarray*}
We add two parity rows below the matrix $A$ to form a new $(n+2) \times k$ matrix 
\begin{eqnarray*}
  \bar{A} = \begin{bmatrix}
    A \\
    A_P
  \end{bmatrix}
\end{eqnarray*}
where
\begin{eqnarray*}
A_P = \begin{bmatrix}
  P_{n1}\boldsymbol{a}_{*1} & P_{n1}\boldsymbol{a}_{*2} & \cdots & P_{n1}\boldsymbol{a}_{*k} \\
  P_{n2}\boldsymbol{a}_{*1} & P_{n2}\boldsymbol{a}_{*2} & \cdots & P_{n2}\boldsymbol{a}_{*k} \\
\end{bmatrix}
\end{eqnarray*}
and 
\begin{eqnarray*}
  P_{n1} &=& (1,1,\ldots,1), \\
  P_{n2} &=& (1,2,\ldots,n).
\end{eqnarray*}
Similarly, we add two parity columns behind the matrix $B$ to form a new $k \times (m+2)$ matrix
\begin{eqnarray*}
  \bar{B} = \begin{bmatrix}
    B & B_P
  \end{bmatrix},
\end{eqnarray*}
where
\begin{eqnarray*}
B_P = \begin{bmatrix}
  P_{m1}\boldsymbol{b}^T_{1*} & P_{m1}\boldsymbol{b}_{2*}^T & \cdots & P_{m1}\boldsymbol{b}_{k*}^T \\
  P_{m2}\boldsymbol{b}_{1*}^T & P_{m2}\boldsymbol{b}_{2*}^T& \cdots & P_{m2}\boldsymbol{b}_{k*}^T \\
\end{bmatrix}^T
\end{eqnarray*}
and 
\begin{eqnarray*}
  P_{m1} &=& (1,1,\ldots,1), \\
  P_{m2} &=& (1,2,\ldots,m).
\end{eqnarray*}
In other words, we use the $(n+2) \times k$ generator matrix
\begin{eqnarray*}
  G_A = \begin{bmatrix}
    I_n \\
    P_{n1} \\
    P_{n2}
  \end{bmatrix}
\end{eqnarray*}
and the $k \times (m+2)$ generator matrix
\begin{eqnarray*}
  G_B = \begin{bmatrix}
    I_m & P^T_{m1} & P^T_{m2}
  \end{bmatrix}
\end{eqnarray*}
to generate the new left matrix $\bar{A}$ and the new right matrix $\bar{B}$, respectively, i.e.,
\begin{eqnarray*}
  \bar{A} = G_AA, \quad \bar{B} = BG_B.
\end{eqnarray*}
Then we compute the matrix with a grid-like structure
\begin{eqnarray*}
  \bar{C} = \bar{A}\bar{B} = 
  \begin{bmatrix}
    I_n \\
    P_{n1} \\
    P_{n2}
  \end{bmatrix}AB \begin{bmatrix}
    I_m & P^T_{m1} & P^T_{m2}
  \end{bmatrix} \triangleq 
  \begin{bmatrix}
    C & P_{r} \\
    P_{c} & P_{g}
  \end{bmatrix},
\end{eqnarray*}
where $P_{r}$ and $P_{c}$ are the parity rows and columns, respectively,
\begin{eqnarray*}
  P_{r} &=& I_n AB \begin{bmatrix}
    P^T_{m1} & P^T_{m2}
  \end{bmatrix} \\
  &=& \begin{bmatrix}
    \boldsymbol{c}_{1*}^T & \boldsymbol{c}_{2*}^T & \cdots & \boldsymbol{c}_{n*}^T
  \end{bmatrix}^T
  \begin{bmatrix}
      P^T_{m1} & P^T_{m2}
  \end{bmatrix}\\
  &=& \begin{bmatrix}
    P_{m1}\boldsymbol{c}^T_{1*} & P_{m1}\boldsymbol{c}_{2*}^T & \cdots & P_{m1}\boldsymbol{c}_{n*}^T \\
    P_{m2}\boldsymbol{c}_{1*}^T & P_{m2}\boldsymbol{c}_{2*}^T & \cdots & P_{m2}\boldsymbol{c}_{n*}^T \\
  \end{bmatrix}^T,
\end{eqnarray*}
\begin{eqnarray*}
  P_{c} &=& \begin{bmatrix}
    P_{n1} \\
    P_{n2}
  \end{bmatrix}AB I_m\\
  &=& \begin{bmatrix}
    P_{n1} \\
    P_{n2}
  \end{bmatrix}
  \begin{bmatrix}
    \boldsymbol{c}_{*1} & \boldsymbol{c}_{*2} & \cdots & \boldsymbol{c}_{*m}
  \end{bmatrix}\\
  &=& \begin{bmatrix}
    P_{n1}\boldsymbol{c}_{*1} & P_{n1}\boldsymbol{c}_{*2} & \cdots & P_{n1}\boldsymbol{c}_{*m} \\
    P_{n2}\boldsymbol{c}_{1*} & P_{n2}\boldsymbol{c}_{2*} & \cdots & P_{n2}\boldsymbol{c}_{m*} \\
  \end{bmatrix},
\end{eqnarray*}
and $P_{g}$ is the global parity,
\begin{eqnarray*}
  P_{g} &=& \begin{bmatrix}
    P_{n1} \\
    P_{n2}
  \end{bmatrix}AB \begin{bmatrix}
    P_{m1} & P_{m2}
  \end{bmatrix}\\
  &=& \begin{bmatrix}
    P_{n1} \\
    P_{n2}
  \end{bmatrix} C \begin{bmatrix}
    P_{m1} & P_{m2}
  \end{bmatrix}\\
  &=& \begin{bmatrix}
    P_{n1}CP_{m1} & P_{n1}CP_{m2} \\
    P_{n2}CP_{m1} & P_{n2}CP_{m2}
  \end{bmatrix}.
\end{eqnarray*}

Now let's analyze the special structure of the output matrix $C$ in detail.

\begin{defn}
For a matrix $\bar{C} = (\bar{c}_{i,j})_{i=1,2,\ldots,n+2}^{j=1,2,\ldots,m+2} \in \mathbb{R}^{(n+2) \times (m+2)}$, we say that the matrix $\bar{C}$ has a {\em grid-like structure} if the following constraints are satisfied:
\begin{enumerate}
  \item Each column of $j \in \{1,2,\ldots,m\}$ satisfies the constraints 
  \begin{eqnarray}\label{cons:1}
    \sum_{i=1}^{n+2} \alpha_{i}^{(\ell)}\bar{c}_{i,j} = 0, \quad \text{for } \alpha^{(\ell)}_{i} \in \mathbb{R}, \ell \in \{1, 2\}.
  \end{eqnarray}

  \item Each row of $i \in \{1,2,\ldots,n\}$ satisfies the constraints
  \begin{eqnarray}\label{cons:2}
    \sum_{j=1}^{m+2} \beta_{j}^{(\ell)}\bar{c}_{i,j} = 0, \quad \text{for } \beta^{(\ell)}_{j} \in \mathbb{R}, \ell \in \{1, 2\}.
  \end{eqnarray}

  \item Each symbol of the matrix $\bar{C}$ satisfies the constraints
  \begin{eqnarray} \label{cons:3}
    \sum_{i=1}^{n+2} \sum_{j=1}^{m+2} \gamma_{i,j}^{(\ell)} \bar{c}_{i,j} &=& 0, \quad \\&&\text{for } \gamma^{(\ell)}_{i,j} \in \mathbb{R}, \ell \in \{1,2,3,4\}. \nonumber
  \end{eqnarray}
\end{enumerate}
\end{defn}

Remark that in our framework the matrix $\bar{C}$ has a grid-like structure and we call it $\mathcal{M}(n+2, k, m+2)$. The framework $\mathcal{M}(n+2, k, m+2)$ has been shown in Fig.~1 intuitively, 
constraint \eqref{cons:1} guarantee that each column $j$ (for $j = 1,\ldots,m$) contains two local parity symbols $\{\bar{c}_{n+1,j}, \bar{c}_{n+2,j}\}$; 
constraint \eqref{cons:2} ensure that every row $i$ (for $i = 1,\ldots,n$) includes two local parity symbols $\{\bar{c}_{i,m+1}, \bar{c}_{i,m+2}\}$;
and constraint \eqref{cons:3} establish four global parity symbols $\{\bar{c}_{n+1,m+1}, \bar{c}_{n+1,m+2}, \bar{c}_{n+2,m+1}, \bar{c}_{n+2,m+2}\}$ at the matrix boundaries.
 
\section{Error correction capability}
In this section, we analyze the error correction capability of our framework $\mathcal{M}(n+2, k, m+2)$. In this paper, we define {\em detect} as identifying the location of errors, whereas {\em correct} refers to rectifying the detected errors.

In our framework, the criterion for judging the error correction capability is not the number of errors that occur, but the 
distribution of the errors. Specifically, how many rows and columns of the matrix $C$ all the errors are distributed in.
Suppose that all errors in matrix $C$ are distributed in $s$ rows and $t$ columns of matrix $C$, where $s \leq n$ and $t \leq m$, and we label this error pattern as E$(s, t)$.
We have the following theorems that can accurately determine the error correction capability of our $\mathcal{M}(n+2,k,m+2)$.

\begin{theorem}
  The framework $\mathcal{M}(n+2, k, m+2)$ can detect and  correct all error pattern E$(s, t)$ (all errors 
  distributed in $s$ rows and $t$ columns) 
  of matrix $C$ if $\max (s, t) \leq 2$.
\end{theorem}
\begin{proof}
    Please see Appendix A.
\end{proof}

Note that the framework $\mathcal{M}(n + 2, k, m + 2)$ can't always detect all error pattern $E(s, t)$ of matrix $C$ if $\max(s,t) > 2 $.
We examine a straightforward example involving two distinct error patterns:
$$\boldsymbol{e}_1 = 
\begin{bmatrix}
1 & -2 & 1 & 0 & \cdots & 0\\
-1 & 2 & -1 & 0 & \cdots & 0\\
0 & 0 & 0 & 0 & \cdots & 0\\
\vdots & \vdots & \vdots & \vdots & \ddots & \vdots \\
0 & 0 & 0 & 0 & \cdots & 0\\
\end{bmatrix}$$
and 
$$\boldsymbol{e}_2 = 
\begin{bmatrix}
0 & 0 & 0 & 0 & \cdots & 0\\
0 & 0 & 0 & 0 & \cdots & 0\\
2 & -4 & 2 & 0 & \cdots & 0\\
-2 & 4 & -2 & 0 & \cdots & 0\\
0 & 0 & 0 & 0 & \cdots & 0\\
\vdots & \vdots & \vdots & \vdots & \ddots & \vdots \\
0 & 0 & 0 & 0 & \cdots & 0\\
\end{bmatrix}.$$
Notably, these two error patterns yield identical local row checksum and column parity symbols, making them indistinguishable through our framework. 

\begin{theorem}
    The framework $\mathcal{M}(n + 2, k, m + 2)$ can always correct all error pattern $E(s, t)$ of matrix $C$ if $\min(s,t) \leq 2$ and the erroneous symbols are known to be localized within specific rows and columns of the result matrix $C$.
\end{theorem}
\begin{proof}
    Please see Appendix B.
\end{proof}

\begin{lemma}
  The framework $\mathcal{M}(n+2, k, m+2)$ can't always correct all error pattern E$(s, t)$ of matrix $C$ if $s=3$ and $t=3$  and the erroneous symbols are known to be localized within specific rows and columns of the result matrix $C$. 
\end{lemma}
\begin{proof}
    Please see Appendix C.
\end{proof}

  \begin{table*}[htbp]
\centering
\begin{tabular}{|c|cccc|cccc|cccc|}
\hline
                       & \multicolumn{4}{c|}{$\delta=0.5$}                                                                    & \multicolumn{4}{c|}{$ \delta=0.1$}                                                                    & \multicolumn{4}{c|}{$ \delta=0.01$}                                                                    \\ \hline
                       & \multicolumn{2}{c|}{$\mathcal{M}(n,k,m)$}                        & \multicolumn{2}{c|}{Checksum}               & \multicolumn{2}{c|}{$\mathcal{M}(n,k,m)$}                        & \multicolumn{2}{c|}{Checksum}               & \multicolumn{2}{c|}{$\mathcal{M}(n,k,m)$}                        & \multicolumn{2}{c|}{Checksum}               \\ \hline
\multicolumn{1}{|c|}{} & \multicolumn{1}{c|}{Rate} & \multicolumn{1}{c|}{Latency} & \multicolumn{1}{c|}{Rate} & Latency & \multicolumn{1}{c|}{Rate} & \multicolumn{1}{c|}{Latency} & \multicolumn{1}{c|}{Rate} & Latency & \multicolumn{1}{c|}{Rate} & \multicolumn{1}{c|}{Latency} & \multicolumn{1}{c|}{Rate} & Latency \\ \hline
a                      & \multicolumn{1}{c|}{$100\%$}      & \multicolumn{1}{c|}{1.28x}        & \multicolumn{1}{c|}{$100\%$}      &    \multicolumn{1}{c|}{1.27x}     & \multicolumn{1}{c|}{$100\%$}      & \multicolumn{1}{c|}{1.28x}        & \multicolumn{1}{c|}{$100\%$}      &     \multicolumn{1}{c|}{1.29x}    & \multicolumn{1}{c|}{$100\%$}      & \multicolumn{1}{c|}{1.29x}        & \multicolumn{1}{c|}{$100\%$}      &       \multicolumn{1}{c|}{1.29x}  \\ \hline
b                     & \multicolumn{1}{c|}{$100\%$}      & \multicolumn{1}{c|}{1.28x}        & \multicolumn{1}{c|}{$100\%$}      &    \multicolumn{1}{c|}{1.28x}     & \multicolumn{1}{c|}{$100\%$}      & \multicolumn{1}{c|}{1.29x}        & \multicolumn{1}{c|}{$100\%$}      &     \multicolumn{1}{c|}{1.28x}    & \multicolumn{1}{c|}{$100\%$}      & \multicolumn{1}{c|}{1.28x}        & \multicolumn{1}{c|}{$100\%$}      &       \multicolumn{1}{c|}{1.29x}  \\ \hline
c                      & \multicolumn{1}{c|}{$100\%$}      & \multicolumn{1}{c|}{1.19x}        & \multicolumn{1}{c|}{$100\%$}      &    \multicolumn{1}{c|}{1.18x}     & \multicolumn{1}{c|}{$100\%$}      & \multicolumn{1}{c|}{1.18x}        & \multicolumn{1}{c|}{$100\%$}      &     \multicolumn{1}{c|}{1.18x}    & \multicolumn{1}{c|}{$100\%$}      & \multicolumn{1}{c|}{1.18x}        & \multicolumn{1}{c|}{$100\%$}      &     \multicolumn{1}{c|}{1.18x}    \\ \hline
d                      & \multicolumn{1}{c|}{$100\%$}      & \multicolumn{1}{c|}{1.36x}        & \multicolumn{1}{c|}{$-$}      &     \multicolumn{1}{c|}{$-$}    & \multicolumn{1}{c|}{$100\%$}      & \multicolumn{1}{c|}{$1.37$x}        & \multicolumn{1}{c|}{$-$}      &    \multicolumn{1}{c|}{$-$}      & \multicolumn{1}{c|}{$100\%$}      & \multicolumn{1}{c|}{$1.36$x}        & \multicolumn{1}{c|}{$-$}      &       \multicolumn{1}{c|}{$-$}   \\ \hline
e                      & \multicolumn{1}{c|}{$100\%$}      & \multicolumn{1}{c|}{1.36x}        & \multicolumn{1}{c|}{$-$}      &     \multicolumn{1}{c|}{$-$}    & \multicolumn{1}{c|}{$100\%$}      & \multicolumn{1}{c|}{$1.37$x}        & \multicolumn{1}{c|}{$-$}      &    \multicolumn{1}{c|}{$-$}      & \multicolumn{1}{c|}{$100\%$}      & \multicolumn{1}{c|}{$1.36$x}        & \multicolumn{1}{c|}{$-$}      &       \multicolumn{1}{c|}{$-$}   \\ \hline
f                      & \multicolumn{1}{c|}{$100\%$}      & \multicolumn{1}{c|}{1.24x}        & \multicolumn{1}{c|}{$-$}      &     \multicolumn{1}{c|}{$-$}    & \multicolumn{1}{c|}{$100\%$}      & \multicolumn{1}{c|}{1.26x}        & \multicolumn{1}{c|}{$-$}      &    \multicolumn{1}{c|}{$-$}     & \multicolumn{1}{c|}{$100\%$}      & \multicolumn{1}{c|}{1.25x}        & \multicolumn{1}{c|}{$-$}      &      \multicolumn{1}{c|}{$-$}   \\ \hline
\end{tabular}
\label{tab:1}
\caption{The simulation result of our framework $M(n=1024, k=4096, m=1024)$ and checksum algorithm.}
\end{table*}

The above lemma tells us that for any error pattern E$(s=3,t=3)$ in $C$, we may not always be able to correct 
it unless we have determined that at least one of the nine positions must be correct so that we can solve the system of equations and get the specific value of each error. Below we will give the theorem with a more universal conclusion.

\begin{corollary}
  The framework $\mathcal{M}(n+2, k, m+2)$ can't always correct all error pattern E$(s, t)$ of matrix $C$ if $\min (s, t) \geq 3$ and the erroneous symbols are known to be localized within specific rows and columns of the result matrix $C$.
\end{corollary}

\begin{proof}
  Note that any error pattern E$(s, t)$ that satisfies $\min(s, t)$ must contain a sub-pattern E$(s=3, t=3)$. According to Lemma 3, this E$(s=3,t=3)$ is not always correctable, so the original error pattern E$(s, t)$ is also uncorrectable.
\end{proof}

\section{Detection and Correcting Process}
In this section, we present the detection and correction process for error pattern E$(s, t)$  of our framework $\mathcal{M}(n+2, k, m+2)$.

For the error pattern E$(s, t)$, we can detect and correct the errors in the matrix $C$ by the following steps:
\begin{enumerate}
  \item \textbf{Detection}. We first detect the error pattern E$(s, t)$ in the matrix $C$. 
  Calculate the sum of the values of the symbols in each row and each column of the matrix $C$ 
  and compare the difference between the corresponding parity symbols and the threshold $\delta$.
  If
  \begin{eqnarray*}
    && |\sum_{j=1}^{m} \bar{c}_{i,j} - \bar{c}_{i, m+1}| \leq \delta, \quad i \in \{1,2,\ldots,n\}, \\
    && |\sum_{j=1}^{m} j\bar{c}_{i,j} - \bar{c}_{i, m+2}| \leq \delta, \quad i \in \{1,2,\ldots,n\}, \\
    && |\sum_{i=1}^{n} \bar{c}_{i,j} - \bar{c}_{n+1, j}| \leq \delta, \quad j \in \{1,2,\ldots,m\},\\
    && |\sum_{i=1}^{n} i\bar{c}_{i,j} - \bar{c}_{n+2, j}| \leq \delta, \quad j \in \{1,2,\ldots,m\},
  \end{eqnarray*}
  then we can determine that the error pattern E$(s, t)$ does not exist in the matrix $C$.
  Otherwise, if there exist two sets of row and column indices $E_r = \{i_1, i_2, \ldots, i_s\}$ and $E_t = \{j_1, j_2, \ldots, j_t\}$ such that
  {\small \begin{align*}
    && |\sum_{j=1}^{m} \bar{c}_{i,j} - \bar{c}_{i, m+1}| > \delta \quad \text{or} \quad |\sum_{j=1}^{m}j \bar{c}_{i,j} - \bar{c}_{i, m+1}| > \delta, \\ &&\qquad i \in E_r, \\
    && |\sum_{i=1}^{n} \bar{c}_{i,j} - \bar{c}_{n+1, j}| > \delta \quad \text{or} \quad |\sum_{i=1}^{n} i\bar{c}_{i,j} - \bar{c}_{n+1, j}| > \delta,  \\&&\qquad j \in E_t.
  \end{align*}}
  Denote $|E_s| = s, |E_t| = t$, if $\min(s, t) \geq 3$, our framework cannot correct errors and needs to be recalculated.
  Else if $\min(s, t) \leq 2$, we can correct the error pattern E$(s, t)$ and go to the next step.

  \item \textbf{Correction}. Without loss of generality, we assume that $s \leq t$.
  If $s = 1$, suppose $E_s = \{s_1\}$, for any $t^* \in E_t$, there must be calculation error occurred in row $s_1$ column $t^*$, then we can 
  correct the symbol by $$\bar{c}^{true}_{s_1, t^*} = \bar{c}_{s_1, t^*} - (\sum_{i=1}^n\bar{c}_{i, t^*} - \bar{c}_{n+1, t^*}).$$ 
  If $s = 2$, suppose $E_s = \{s_1, s_2\}$, for any $t^* \in E_t$, there may be calculation error occurred in row $s_1$ or $s_2$ column $t^*$,
  then we have an equation system
  \begin{eqnarray*}
    \begin{bmatrix}
      1 & 1 \\
      s_1 & s_2
    \end{bmatrix}
    \begin{bmatrix}
      e_{s_1, t^*} \\
      e_{s_2, t^*}
    \end{bmatrix} = 
    \begin{bmatrix}
      \bar{c}_{n+1, t^*} - P_{n1}\boldsymbol{c}_{*t^*} \\
      \bar{c}_{n+2, t^*} - P_{n2}\boldsymbol{c}_{*t^*}
    \end{bmatrix}.
  \end{eqnarray*}
  Solve this equation system, we can get the correct symbols 
  \begin{eqnarray*}
    &&\bar{c}^{true}_{s_1, t^*} = \bar{c}_{s_1, t^*} - e_{s_1, t^*}, \\
    &&\bar{c}^{true}_{s_2, t^*} = \bar{c}_{s_2, t^*} - e_{s_2, t^*}.
  \end{eqnarray*}
\end{enumerate}
Our detection and correction process can detect and correct the error pattern E$(s, t)$ in the matrix $C$. The corrected symbols may exhibit minor deviations from their true values, which can be attributed solely to computational rounding errors. These discrepancies remain within acceptable bounds and do not significantly impact the overall correction accuracy.

\section{Simulation Experiments and Analysis}
In this section, we present the experimental results of our framework $\mathcal{M}(n+2, k, m+2)$ and the checksum algorithm \cite{braun2014abft} for matrix multiplication.

In our experiments, we consider conducting error correction tests for the matrix multiplication on V100 GPU across multiple scenarios.


We first focus on the error-correcting capabilities of our framework when a single error occurs in matrix multiplication. The main scenarios we consider are as follows: (a) a symbol error in the left matrix $A$, (b) a 
symbol error in the right matrix $B$, and (c) a symbol error in the resulting matrix $C$.
To demonstrate the error-correction capabilities of our framework, we mainly consider three scenarios: (d) one symbol error each in the left 
matrix $A$ and the output matrix $C$, (e) one symbol error each in the right matrix $B$ and the output matrix $C$, and (f) two symbol errors occurring in the output matrix $C$.

Our experiments are shown in Table 1. For the first three types of silent data corruption, our algorithm maintains roughly the same performance as the checksum algorithm, both being able to correctly correct errors with 100\% probability at a cost of 18\% to 29\%. For the last three types of silent data corruption, the checksum algorithm cannot perform normal error correction because the error type is beyond the checksum algorithm's ability, but our algorithm can correctly correct all errors with 100\% probability at a cost of 24\% to 37\%. Especially, our algorithm can correct any two errors in the output matrix C with probability 100\% at a cost of only 20\%.

\section{Conclusion}
In this paper, we propose a novel error-correcting coding framework for matrix multiplication. We have fully explored the algorithm's error correction capability and used a theoretical proof to provide an upper limit for the algorithm's error correction capability. Through experiments, our algorithm not only performs on par with existing algorithms within the error range that existing algorithms can support, but also provides a highly reliable and efficient correction strategy in areas where existing algorithms cannot correctly correct errors. How to further reduce the cost of efficient correction of matrix multiplication is one of our future directions.

\ifCLASSOPTIONcaptionsoff
  \newpage
\fi
\newpage
\bibliographystyle{IEEEtran}
\bibliography{main}

\clearpage
\appendix
\subsection{Proof of Theorem 1}
\begin{proof}
   Consider an error pattern where all erroneous symbols are confined to a subset of $s$ rows and $t$ columns in matrix $C$ with $\max (s, t) \leq 2$. 
  If $s=t=1$ (i.e., only a single symbol is erroneous), the pattern becomes trivial, as the error can be corrected using a conventional checksum. Therefore, we focus on cases where two or more errors occur, meaning at least one of $s$ or $t$ must be 2. Without loss of generality, we assume $t=2$.
  We consider the following two cases:
  \begin{enumerate}
    \item $s = 1$. In this case, all errors are localized within a single row of the matrix $C$. Error detection is performed by comparing two distinct metrics for each row: (1) the difference between the row sum and the first row checksum, and (2) the difference between the weighted row sum and the second row checksum. If either of these differences exceeds a predetermined threshold $\delta$ for any given row, it signifies the presence of erroneous symbols, thereby enabling the identification of the faulty row, denoted as $s_1$.
    The column checksum symbols $\{\bar{c}_{n+1, j}\}_{j=1,2,\ldots,m}$ are then utilized in a manner analogous to the row checksum verification process. This allows for the precise determination of the column indices $\{j_1,j_2\}$ containing the erroneous symbols.
    As stipulated by the constraints in (\ref{cons:2}), the errors located in row $s_1$ can be effectively corrected through the application of the corresponding column parity symbols $\{\bar{c}_{n+1,j_1}, \bar{c}_{n+1, j_2}\}$. 
    
    \item $s = 2$. In this case, all errors are distributed in two rows of matrix $C$. Error detection is achieved by evaluating two independent metrics for each row: (1) the discrepancy between the row sum and the first row checksum, and (2) the discrepancy between the weighted row sum and the second row checksum. If either discrepancy exceeds a predefined threshold $\delta$ for any row, it confirms the presence of erroneous symbols, thereby enabling the identification of the faulty rows $\{s_1, s_2\}$. Similarly by using the checksum symbols $\{\bar{c}_{m+1, j}\}_{j=1,2,\ldots,m}$ for each column, we can determine the specific column set $\{j_1,j_2\}$ in which the error occurred.
    According to the constraints (\ref{cons:2}), we can obtain that the errors in the two rows can be corrected by the column parity symbols $\{\bar{c}_{n+1, j_1}, \bar{c}_{n+2, j_1}, \bar{c}_{n+1, j_2}, \bar{c}_{n+2, j_2}\}$.
  \end{enumerate}
  The local parity symbols can correct the errors in both cases. 
  Therefore, the theorem holds.
\end{proof}

  \begin{figure*}[htpb]
\begin{equation}\label{eq:5}
\begin{bmatrix}
      1 & 1 & 1 & 0 & 0 & 0 & 0 & 0 & 0 \\
      t_1 & t_2 & t_3 & 0 & 0 & 0 & 0 & 0 & 0 \\
      0 & 0 & 0 & 1 & 1 & 1 & 0 & 0 & 0 \\
      0 & 0 & 0 & t_1 & t_2 & t_3 & 0 & 0 & 0 \\
      0 & 0 & 0 & 0 & 0 & 0 & 1 & 1 & 1 \\
      0 & 0 & 0 & 0 & 0 & 0 & t_1 & t_2 & t_3 \\
      s_1 & 0 & 0 & s_2 & 0 & 0 & s_3 & 0 & 0 \\
      0 & s_1 & 0 & 0 & s_2 & 0 & 0 & s_3 & 0 \\
      0 & 0 & s_1 & 0 & 0 & s_2 & 0 & 0 & s_3 \\
      1 & 0 & 0 & 1 & 0 & 0 & 1 & 0 & 0 \\
      0 & 1 & 0 & 0 & 1 & 0 & 0 & 1 & 0 \\
      0 & 0 & 1 & 0 & 0 & 1 & 0 & 0 & 1 \\
      1 & 1 & 1 & 1 & 1 & 1 & 1 & 1 & 1 \\
      t_1 & t_2 & t_3 & t_1 & t_2 & t_3 & t_1 & t_2 & t_3 \\
      s_1 & s_2 & s_3 & s_1 & s_2 & s_3 & s_1 & s_2 & s_3 \\
      t_1 s_1 & t_2 s_1 & t_3 s_1 & t_1 s_2 & t_2 s_2 & t_3 s_2 & t_1 s_3 & t_2 s_3 & t_3 s_3
      \end{bmatrix}\begin{bmatrix}
      e_{s_1, t_1} \\
      e_{s_1, t_2} \\
      e_{s_1, t_3} \\
      e_{s_2, t_1} \\
      e_{s_2, t_2} \\
      e_{s_2, t_3} \\
      e_{s_3, t_1} \\
      e_{s_3, t_2} \\
      e_{s_3, t_3}
      \end{bmatrix}=
      \begin{bmatrix}
        P_{m1}\boldsymbol{c}_{s_1*} - \bar{c}_{s_1, m+1} \\
        P_{m2}\boldsymbol{c}_{s_1*} - \bar{c}_{s_1, m+2} \\
        P_{m1}\boldsymbol{c}_{s_2*} - \bar{c}_{s_2, m+1} \\
        P_{m2}\boldsymbol{c}_{s_2*} - \bar{c}_{s_2, m+2} \\
        P_{m1}\boldsymbol{c}_{s_3*} - \bar{c}_{s_3, m+1} \\
        P_{m2}\boldsymbol{c}_{s_3*} - \bar{c}_{s_3, m+2} \\
        P_{n1}\boldsymbol{c}_{*t_1} - \bar{c}_{n+1, t_1} \\
        P_{n2}\boldsymbol{c}_{*t_1} - \bar{c}_{n+2, t_1} \\
        P_{n1}\boldsymbol{c}_{*t_2} - \bar{c}_{n+1, t_2} \\
        P_{n2}\boldsymbol{c}_{*t_2} - \bar{c}_{n+2, t_2} \\
        P_{n1}\boldsymbol{c}_{*t_3} - \bar{c}_{n+1, t_3} \\
        P_{n2}\boldsymbol{c}_{*t_3} - \bar{c}_{n+2, t_3} \\
        P_{n1}CP_{m1} - \bar{c}_{n+1, m+1} \\
        P_{n1}CP_{m2} - \bar{c}_{n+1, m+2} \\
        P_{n2}CP_{m1} - \bar{c}_{n+2, m+1} \\
        P_{n2}CP_{m2} - \bar{c}_{n+2, m+2} \\
      \end{bmatrix}
\end{equation}
\end{figure*}
\subsection{Proof of Theorem 2}
\begin{proof}
   Suppose that there are $s$ rows and $t$ columns of matrix $C$ all the errors are distributed in, where $s \leq n$ and $t \leq m$, and $\min (s, t) \leq 2$. 
  Without loss of generality, we assume that $s \leq t$, i.e., $s \leq 2$. 
  We consider the following two cases:
  \begin{enumerate}
    \item $s = 1$. In this case, all errors are distributed in one row (assume row $s_1$) of matrix $C$.
    According to the constraint (\ref{cons:2}), we can obtain that the errors in the row $s_1$ can be corrected by the column parity symbols $\{\bar{c}_{n+1,j_1}, \bar{c}_{n+1, j_2},\ldots, \bar{c}_{n+1, j_t}\}$.
    \item $s = 2$. In this case, all errors are distributed in two rows (assume row $s_1$ and $s_2$) of matrix $C$.
    According to the constraint (\ref{cons:2}), we can obtain that the errors in the two rows can be corrected by the column parity symbols $\{\bar{c}_{n+1, j_1}, \bar{c}_{n+2, j_1}, \bar{c}_{n+1, j_2}, \bar{c}_{n+2, j_2}, \ldots, \bar{c}_{n+1, j_t}, \bar{c}_{n+2, j_t}\}$.
  \end{enumerate}
  In both cases, the errors can be corrected by the local parity symbols. 
  Therefore, the theorem holds.
\end{proof}

\subsection{Proof of Lemma 3}
\begin{proof}
  If $s=3$ and $t=3$, we consider the error pattern E$(s = 3, t = 3)$: 
  all errors are distributed in three rows $s_1, s_2, s_3$ and three columns $t_1, t_2, t_3$ of matrix $C$
  where $1\leq s_1 < s_2 < s_3 \leq n$ and $1\leq t_1 < t_2 < t_3 \leq m$.

  Notice that there are a total of nine positions in the three rows and three columns where errors may occur. 
  We assume that the difference between the error symbols and the correct symbols at the nine positions are
  \begin{eqnarray*}
    e_{s_1, t_1}, e_{s_1, t_2}, e_{s_1, t_3}, e_{s_2, t_1}, e_{s_2, t_2}, e_{s_2, t_3}, e_{s_3, t_1}, e_{s_3, t_2}, e_{s_3, t_3}.
  \end{eqnarray*}
  In our $\mathcal{M}(n+2, k, m+2)$, 
  there are 6 row-parity symbols, 6 column-parity symbols, and 4 global parity symbols associated with these 9 symbols, for a total of 16 constraints.
  We can write the 16 constraints in Eq. \eqref{eq:5},
  where the coefficient matrix is a $16 \times 9$ matrix and the right-hand side vector is a $16 \times 1$ vector.

  Using the Gauss elimination method, we perform row elimination on the coefficient matrix and obtain the matrix after row elimination
  \begin{eqnarray*}
    \begin{bmatrix}
      1 & 0 & 0 & 0 & 0 & 0 & 0 & 0 & \frac{-s_2 t_2 + s_2 t_3 + s_3 t_2 - s_3 t_3}{s_1 t_1 - s_1 t_2 - s_2 t_1 + s_2 t_2} \\
      0 & 1 & 0 & 0 & 0 & 0 & 0 & 0 & \frac{s_2 t_1 - s_2 t_3 - s_3 t_1 + s_3 t_3}{s_1 t_1 - s_1 t_2 - s_2 t_1 + s_2 t_2} \\
      0 & 0 & 1 & 0 & 0 & 0 & 0 & 0 & \frac{-s_2 + s_3}{s_1 - s_2} \\
      0 & 0 & 0 & 1 & 0 & 0 & 0 & 0 & \frac{s_1 t_2 - s_1 t_3 - s_3 t_2 + s_3 t_3}{s_1 t_1 - s_1 t_2 - s_2 t_1 + s_2 t_2} \\
      0 & 0 & 0 & 0 & 1 & 0 & 0 & 0 & \frac{-s_1 t_1 + s_1 t_3 + s_3 t_1 - s_3 t_3}{s_1 t_1 - s_1 t_2 - s_2 t_1 + s_2 t_2} \\
      0 & 0 & 0 & 0 & 0 & 1 & 0 & 0 & \frac{s_1 - s_3}{s_1 - s_2} \\
      0 & 0 & 0 & 0 & 0 & 0 & 1 & 0 & \frac{t_2 - t_3}{-t_1 + t_2} \\
      0 & 0 & 0 & 0 & 0 & 0 & 0 & 1 & \frac{-t_1 + t_3}{-t_1 + t_2} \\
      0 & 0 & 0 & 0 & 0 & 0 & 0 & 0 & 0 \\
      \vdots & \vdots & \vdots & \vdots & \vdots & \vdots & \vdots & \vdots & \vdots \\
      0 & 0 & 0 & 0 & 0 & 0 & 0 & 0 & 0 \\
      \end{bmatrix}.
  \end{eqnarray*}
  From this we can see that no matter what values $s_1,s_2,s_3,t_1,t_2,t_3$ take, the matrix after row elimination has only 8 non-zero rows, so the rank of the matrix is 8, 
  that is, the rank of the original coefficient matrix is $8 < 9$, and the error pattern cannot be corrected in this case.
\end{proof}

\end{document}